\documentclass[11pt,a4paper,reqno]{amsart}

\usepackage{amsopn}
\usepackage{amsmath}
\usepackage{amssymb}
\usepackage{amsthm}
\usepackage{hyperref}
\usepackage{graphicx}
\usepackage{color}
\usepackage{listings}
\usepackage{mathtools}
\usepackage{subcaption}
\usepackage{enumerate}
\usepackage{cite}
\usepackage{todonotes}
\usepackage[foot]{amsaddr}
\usepackage{amsfonts}

\usepackage{multirow}
\usepackage{adjustbox}
\usepackage{geometry}

\usepackage{tikz}
\usetikzlibrary{angles,quotes}
\usetikzlibrary{decorations.pathreplacing}

\DeclareMathOperator{\tr}{tr}
\DeclareMathOperator{\Tr}{Tr}

\newcommand{\R}{\ensuremath{\mathbb{R}}}

\newcommand{\C}{\ensuremath{\mathbb{C}}}

\newcommand{\ket}[1]{\ensuremath{|#1\rangle}}
\newcommand{\bra}[1]{\ensuremath{\langle#1|}}
\newcommand{\ketbra}[2]{\ensuremath{\ket{#1} \! \bra{#2}}}
\newcommand{\proj}[1]{\ensuremath{\ketbra{#1}{#1}}}
\newcommand{\braket}[2]{\ensuremath{\langle{#1}|{#2}\rangle}}

\newcommand{\Id}{{\rm 1\hspace{-0.9mm}l}}

\newcommand{\XX}{\mathcal{X}}

\newcommand{\YY}{\mathcal{Y}}
\newcommand{\ZZ}{\mathcal{Z}}

\newcommand{\s}{\mathcal{S}}

\newcommand{\CC}{\mathcal{C}}

\newcommand{\A}{\mathcal{A}}
\newcommand{\B}{\mathcal{B}}

\newcommand{\VV}{\mathcal{V}}

\newtheorem{definition}{Definition}

\newtheorem{theorem}{Theorem}

\newtheorem{corollary}{Corollary}
\newtheorem{proposition}{Proposition}

\newtheorem*{rem*}{Remark}

\usepackage{caption}
\def\>{\rangle}
\def\<{\langle}

\title{Strategies for single-shot discrimination of process matrices}

\author{Paulina Lewandowska$^{1}$}
\author{Łukasz Pawela$^{1}$}
\author{Zbigniew Puchała$^{1}$}

\address{$^1$Institute of Theoretical and Applied Informatics, Polish 
Academy of Sciences, ul. Ba{\l}tycka 5, 44-100 Gliwice, Poland}

\begin{document}
\maketitle

\begin{abstract}
	The topic of causality has recently gained traction quantum information research.
This work examines the problem of single-shot discrimination between process matrices which are an universal method defining a causal structure.
 We provide an exact expression for the optimal probability of correct distinction.  In addition, we present an alternative way to achieve this expression by using the convex cone structure theory. We also express the discrimination task as semidefinite programming. 
Due to that, we have created the SDP calculating the distance between process matrices and we quantify it in terms of the trace norm. As a valuable by-product, the program finds an optimal realization of the discrimination task. We also find two classes of process matrices which can be distinguished perfectly. Our main result, however, is a consideration of the discrimination task for process matrices corresponding to quantum combs.  We study which strategy, adaptive or non-signalling, 
should be used during the discrimination task. 
We proved that no matter which strategy you choose, the probability of distinguishing two process matrices being a quantum comb is the same. 
\end{abstract}	

\section{Introduction}
 The topic of causality has remained a staple in quantum physics and quantum information theory for recent years.
The idea of a causal influence in quantum physics is best illustrated by considering two characters, Alice and Bob, preparing experiments in two separate laboratories.  Each of them receives a physical system and performs an operation on it. After that, they send their respective system out of the laboratory. In a causally ordered framework, there are three possibilities: Bob cannot signal to Alice, which means the choice of Bob's action cannot influence the statistics Alice records (denoted by $A \prec B$), Alice cannot signal to Bob  ($B \prec A$), or neither party can influence the other $(A || B)$.   A causally neutral formulation of quantum theory is described in terms of  quantum   combs~\cite{bisio2011quantum}.

One may wonder  if Alice's and Bob's action can influence each other. It might seem impossible, except in a world with closed time-like curves (CTCs)~\cite{godel1949example}.  But the existence of CTCs  implies some logical paradoxes, such as the  grandfather paradox~\cite{deutsch1994quantum}. Possible solutions have been proposed in which quantum mechanics and CTCs can exist and such paradoxes are avoided, but  modifying quantum theory into a nonlinear one~\cite{gisin1990weinberg}. A natural question arises: is it possible to keep the framework of linear quantum theory and still go beyond definite causal structures?

One such framework was proposed  by Oreshkov, Costa and Brukner ~\cite{oreshkov2012quantum}. They introduced  a new resource called  a process matrix -- a generalization of the notion of quantum state. This new approach has 
provided a consistent  representation of correlations in casually and non-causally related experiments.  Most interestingly, they have described a situation that two actions are neither causally ordered and one cannot say which 
action influences the second one.
Thanks to that, the term of causally non-separable (CNS)  structures started to correspond to superpositions of situations in which, roughly speaking, Alice can signal to Bob, and Bob can signal to Alice, jointly. A general overview of causal connection theory is described in \cite{brukner2014quantum}.

The indefinite causal structures could make a new aspect of quantum information
processing. 
This more general model of computation can outperform causal quantum computers in specific tasks, such as  learning  or discriminating  between two quantum channels \cite{bavaresco2021strict, quintino2022deterministic, bavaresco2022unitary}.   
The problem of discriminating quantum operations
is of the utmost importance in modern quantum information science.
Imagine we have an unknown operation hidden in a black box. We only have information that it is one of two operations. The goal is to determine an optimal strategy for this process that achieves the
highest possible probability of discrimination.
 For the case of  a single-shot discrimination scenario, researchers  have used  different approaches, with the possibility of using entanglement in order to perform an optimal protocol.
In~\cite{duan2007entanglement}, Authors have shown that in the task of discrimination of unitary channels, the entanglement is not necessary, whereas for quantum measurements \cite{d2001using, cao2015perfect, puchala2018strategies},  we need to use entanglement.  
 Considering   multiple-shot  discrimination scenarios, researchers  have utilized   parallel or adaptive approaches. 
   In the parallel case, 
  we establish that the
  discrimination between operations does not require pre-processing and post-processing. One example of such an approach is distinguishing unitary channels~\cite{duan2007entanglement}, or von Neumann measurements \cite{puchala2021multiple}. 
 The case when the black box can be
 used multiple times in an adaptive way was investigated by the authors of \cite{ wang2006unambiguous, krawiec2020discrimination},  who have proven that the use of adaptive strategy  and a general  notion of quantum combs
 can improve  discrimination.

 In this work, we study the problem of discriminating process matrices in a single-shot scenario. 
We obtain that the  probability of correct distinction  process matrices is strictly related to the Holevo-Helstrom theorem for quantum channels. 
Additionally, we  write this result as
a semidefinite program (SDP) which is numerically efficient. The SDP program allows us to find
an optimal discrimination strategy. 
We compare the effectiveness of the obtained strategy with the previously mentioned strategies.
The problem gets more complex in the case when we consider the non-causally ordered framework.  In this case, we consider the discrimination task between two process matrices having different causal orders.

This paper is organized as follows. In Section~\ref{sec-preliminaries}
 we introduce  necessary mathematical
framework. Section~\ref{sec-process-matrices} is dedicated to the concept of process matrices.  Section~\ref{sec-discrimination-prob} presents the discrimination task  between pairs of process matrices and calculate the exact probability of distinguishing them. 
Some examples of discrimination between  different classes of process matrices are presented in Section~\ref{sec-discrimination-example}. 
In Section~\ref{sec-free}, we consider the discrimination task between free process matrices, whereas in 
Section~\ref{sec-comb} we consider the discrimination task between process matrices being quantum combs. In 
Section~\ref{sec-perfect-discrim}, we show a particular class of process matrices having opposite causal structures which can be distinguished perfectly.
Finally, Section~\ref{spd-prob} and Section~\ref{sec-sdp-distance}
 are devoted to semidefinite programming, thanks to which, among other things, we obtain an optimal discrimination strategy.
In Section~\ref{sec-convex}, we analyze an alternative way to achieve this expression  using the convex cone structure theory.
Concluding remarks are presented in
the final Section~\ref{sec-conclusion}.
In the Appendix~\ref{app:convex}, we
provide technical details about the convex cone structure.

\section{Mathematical preliminaries }\label{sec-preliminaries}

Let us introduce the following notation. Consider two complex Euclidean spaces and denote them by  $\XX, \YY$. 
By  $\mathrm{L}(\XX , \YY)$ we denote 
the collection of all linear mappings of the form 
$A: \XX \rightarrow \YY$. 
As a shorthand put
$\mathrm{L}(\XX) \coloneqq \mathrm{L}(\XX, \XX).$ 
By $\mathrm{Herm}(\XX)$ we denote 
the set of Hermitian operators while the subset of $\mathrm{Herm}(\XX)$ consisting of  positive semidefinite
operators will be denoted by    $\mathrm{Pos}(\XX)$.   The set of
quantum states, that is positive semidefinite operators  $\rho$
such that $\tr\rho = 1$, will be denoted by $\Omega(\XX)$.
 An operator $U \in \mathrm{L}\left(\XX\right) $ is unitary if  it  satisfies
the equation $U U^\dagger = U^\dagger U = \Id_\XX$. 
The notation $\mathrm{U}\left(\XX\right)$ will be used to denote the set of  all unitary operators.
We will also need a linear
mapping of the form $\Phi: \mathrm{L}(\XX) \rightarrow \mathrm{L}(\YY)$ transforming $\text{L}(\XX)$ into $\text{L}(\YY)$.
The set of all linear mappings is denoted $\mathrm{M}(\XX, \YY)$. 
There
exists a bijection between set  $\mathrm{M}(\XX, \YY)$ 
and the set of operators $\mathrm{L}(\YY \otimes \XX)$ known as the 
Choi~\cite{choi1975completely} and Jamio{\l}kowski~\cite{jamiolkowski1972linear} isomorphism.
For a given linear mapping $\Phi_M:  \mathrm{L}(\XX) \rightarrow \mathrm{L}(\YY)$ corresponding  Choi matrix $M \in \mathrm{L}(\YY \otimes \XX) $ can be  explicitly written as 
\begin{equation}
M \coloneqq \sum_{i,j = 0}^{\dim(\XX) - 1} \Phi_M(\ketbra{i}{j}) \otimes \ketbra{i}{j}. \end{equation}
We will denote linear
mappings by $\Phi_M, \Phi_N, \Phi_R$ etc., whereas the
corresponding Choi matrices as plain symbols: $M, N, R$ etc. Let us consider a composition of mappings $\Phi_R = \Phi_N \circ \Phi_M$ where   $\Phi_N: \mathrm{L}(\ZZ) \rightarrow  \mathrm{L}(\YY)$ and $\Phi_M: \mathrm{L}(\XX) \rightarrow  \mathrm{L}(\ZZ)$ with Choi matrices $N \in \mathrm{L}(\ZZ \otimes \YY)$ and $M \in  \mathrm{L}(\XX \otimes \ZZ)$, respectively. 
Then, the Choi matrix of  $\Phi_R $  is given by~\cite{chiribella2009theoretical}
\begin{equation}
R = \tr_{\ZZ} \left[ \left(\Id_\YY \otimes M^{T_\ZZ}\right)\left( N \otimes \Id_\XX \right) \right], 
\end{equation}
where $M^{T_\ZZ} $ denotes the partial transposition of $M$ on the subspace $\ZZ$. 
 The above result can be expressed 
by introducing the notation of the link product of the operators 
$N$ and $M $ as \begin{equation}
N * M \coloneqq  \tr_{\ZZ} \left[ \left(\Id_\YY \otimes M^{T_\ZZ}\right)\left( N \otimes \Id_\XX \right) \right]. 
\end{equation}

Finally, we introduce a special subset of all mappings $\Phi$, called quantum channels, which are  completely positive
and trace preserving (CPTP).
In other words, the first condition
reads \begin{equation}
 (\Phi \otimes \mathcal{I_\ZZ})(X) \in   \mathrm{Pos}(\YY \otimes \ZZ)
\end{equation}
for all $ X \in  \mathrm{Pos}(\XX \otimes \ZZ) $ and $\mathcal{I_\ZZ}$ is an identity channel acts on  $\mathrm{L}(\ZZ)$ for any $\ZZ$, while the second condition reads
\begin{equation}
\tr(\Phi(X)) = \tr(X) 
\end{equation}
for all $X \in \mathrm{L}(\XX)$.

In this work we will consider a special class of quantum channels called  non-signaling channels (or causal channels)~~\cite{beckman2001causal, piani2006properties}. We say that $\Phi_N: \mathrm{L}(\XX_I \otimes \YY_I) \rightarrow \mathrm{L}(\XX_O \otimes \YY_O)$ 
is a non-signaling channel if its Choi operator satisfies the following conditions 
\begin{equation}
\begin{split}
&\tr_{\XX_O} (N) = \frac{\Id_{\XX_I}}{\dim(\XX_I)} \otimes \tr_{\XX_O\XX_1} (N), \\ &\tr_{\YY_O} (N) = \frac{\Id_{\YY_I}}{\dim(\YY_I)} \otimes \tr_{\YY_O\YY_1} (N).
\end{split}
\end{equation}	
It can be shown 
\cite{chiribella2013quantum} 
that each non-signaling channel is an affine combination
of product channels. More precisely, any non-signaling  channel $\Phi_N : \mathrm{L}(\XX_I \otimes \YY_I) \rightarrow \mathrm{L}(\XX_O \otimes \YY_O) $ can  be written as 
\begin{equation}
\Phi_N = \sum_i \lambda_i \Phi_{S_i} \otimes \Phi_{T_i} ,
\end{equation}
where  $\Phi_{S_i} : \mathrm{L}(\XX_I ) \rightarrow \mathrm{L}(\XX_O ) $ and $\Phi_{T_i}:  \mathrm{L}(\YY_I) \rightarrow \mathrm{L}( \YY_O) $ are quantum channels, $\lambda_i \in \R$ such that $\sum_i \lambda_i = 1$.  
For the rest of this paper, by $\mathbf{NS}(\XX_I \otimes \XX_O \otimes \YY_I \otimes \YY_O)$ we will denote 
the set of Choi matrices of non-signaling channels. 

The most general quantum
operations are represented by quantum instruments~\cite{nielsen2010quantum, davies1970operational},  that
is, collections of completely positive (CP) maps $\left\{ \Phi_{M_i} \right\}_i$  associated to all measurement outcomes, characterized by the property that $\sum_i \Phi_{M_i}$ is a quantum channel. 

We will also consider the concept of quantum network and tester~\cite{chiribella2008quantum}.
We say that  $\Phi_{R^{(N)}}$ is a deterministic quantum network (or quantum comb) if it is a concatenation of $N$ quantum channels and $
R^{(N)} \in \mathrm{L} \left( \bigotimes_{i=0}^{2N-1}  \XX_i \right)$
fulfills the following conditions
%
\begin{equation}
\begin{split}
  R^{(N)} & \ge 0,  \\
\tr_{\XX_{2k-1}} \left( R^{(k)} \right) & =  \Id_{\XX_{2k-2}} \otimes R^{(k-1)}, 
\end{split}
\end{equation}
where $R^{(k-1)} \in \mathrm{L} \left( \bigotimes_{i=0}^{2k-3}  \XX_i \right) $ is the Choi matrix of the reduced quantum comb with concatenation of $k-1$ quantum channels, $k = 2,\ldots, N$.  
We remind that a probabilistic quantum
network $\Phi_{S^{(N)}}$ is equivalent to a concatenation of $N$ completely positive trace non increasing
linear maps. Then, the Choi operator $S^{(N)}$ of $\Phi_{S^{(N)}}$ satisfies $0 \le S^{(N)} \le R^{(N)}$, where  $R^{(N)}$ is Choi matrix of a quantum comb. Finally, we recall the definition of a quantum tester.  A quantum tester is a collection  of  probabilistic quantum networks $\left\{ R_{i}^{(N)} \right\}_i$ whose sum is a quantum comb, that is $\sum_i R^{(N)}_i = R^{(N)}$, and additionally $\dim(\XX_0) = \dim(\XX_{2N-1}) = 1$.

	We will also use the Moore--Penrose pseudo--inverse by abusing notation $X^{-1} \in \mathrm{L}(\YY, \XX)$ for an operator $X 
	\in \mathrm{L}(\XX, \YY)$. Moreover, we introduce the vectorization operation of  $X$ defined by $|X \rangle \rangle = \sum_{i=0}^{\dim(\XX) -1} (X \ket{i}) \otimes \ket{i} $.

 \section{Process matrices}\label{sec-process-matrices}

This section introduces the formal definition of the process matrix with its characterization and intuition. 
Next, we present some classes of process matrices considered in this paper. 

Let us define the operator  $\prescript{}{\XX}{Y}$ as 
\begin{equation}
\prescript{}{\XX}{Y}  = \frac{\Id_\XX}{\dim(\XX)} \otimes \tr_\XX Y
\end{equation} for every $Y \in \mathrm{L}(\XX \otimes \ZZ)$, where $\ZZ$ is an arbitrary complex Euclidean space. 
We will also need the following projection operator \begin{equation}\label{proj}
L_V(W) =  \prescript{}{\A_O}{W} +  \prescript{}{\B_O}{W} - \prescript{}{\A_O\B_O}{W} - \prescript{}{\B_I\B_O}{W}  + \prescript{}{\A_O\B_I\B_O}{W} - \prescript{}{\A_I\A_O}{W} + \prescript{}{\A_O\A_I\B_O}{W}.
\end{equation} where $W \in \mathrm{Herm}(\A_I\otimes \A_O \otimes \B_I \otimes \B_O)$.

\begin{definition}\label{process-matrix-1}
	We say that $W \in \mathrm{Herm}(\A_I\otimes \A_O \otimes \B_I \otimes \B_O) $ is  a process matrix if it fulfills the following conditions
	\begin{equation}
	W \ge 0, \,\,\, W  = L_V(W), \,\,\ \tr(W) = \dim(\A_O) \cdot \dim(\B_O),
	\end{equation}
	where the projection operator $L_V$ 
	is defined by Eq.~\eqref{proj}.

\end{definition}
The set of all  process matrices  will be denoted by $\mathbf{W^{PROC}}$. In the upcoming  considerations, it will be more convenient to work with the equivalent characterization of 
process matrices which
can be found in~\cite{araujo2015witnessing}. 
\begin{definition}\label{process-matrix-2}
	We say that $W \in \mathbf{W^{PROC}}$ is a process matrix if it fulfills the following conditions
	\begin{equation*}
	W   \ge 0, 
	\end{equation*}
	\begin{equation*}
	\prescript{}{\A_I\A_O}{W} = \prescript{}{\A_O\A_I\B_O}{W}, 
	\end{equation*}
	\begin{equation}
	\prescript{}{\B_I\B_O}{W} = \prescript{}{\A_O\B_I\B_O}{W},
	\end{equation}
	\begin{equation*}
	W = \prescript{}{\B_O}{W} +\prescript{}{\A_O}{W} - \prescript{}{\A_O\B_O}{W},	\end{equation*}
	\begin{equation*}	\tr (W) = \dim(\A_O) \cdot \dim(\B_O).
	\end{equation*}
\end{definition}

The   concept of  process matrix can be best  illustrated by considering two characters, Alice   and Bob, performing experiments in two separate laboratories. 
Each party  acts in a local laboratory, which can be identified by an input  space $\A_I$ and an output space $\A_O$ for Alice, and analogously $\B_I $ and $\B_O$ for Bob.
In general, a label $i$, denoting  Alice's measurement outcome, is associated with the CP map   $\Phi_{M^{A}_i} $ obtained from the instrument $\left\{ \Phi_{M^{A}_i}  \right\}_i$. 
Analogously, the Bob's measurement outcome $j$ is associated with the map $\Phi_{M^{B}_j}$ from the instrument  $\left\{ \Phi_{M^{B}_j}  \right\}_j$. Finally, the joint
probability for a pair of outcomes $i$ and $j$ 
can be expressed as
\begin{equation}\label{eq-prob}
p_{ij} = \tr \left[ W \left(M_i^A \otimes  M_j^B \right)  \right],
\end{equation}
where $W \in \mathbf{W^{PROC}}$ is a process matrix  that describes the causal
structure outside of the laboratories. 
The  valid process
matrix is defined by the requirement that probabilities are
well defined, that is, they must be non-negative and sum
up to one. These requirements give us the conditions present in Definition~\ref{process-matrix-1}
and Definition~\ref{process-matrix-2}.

In the general case, the Alice's and Bob's strategies can be more complex than 
the product strategy $M_i^A \otimes M_j^B$ which defines the probability 
$p_{ij}$ given by Eq.~\eqref{eq-prob}. If their action is somehow correlated, 
we can write the associated instrument in the following form $\left\{ \Phi_{ 
N_{ij}^{AB}}\right\}$. 
 It was observed in~\cite{araujo2015witnessing} that this instrument describes a valid strategy, that is
 \begin{equation}\label{eq:nonsignalling}
 \tr \left( W \sum_{ij} N_{ij}^{AB} \right) = 1 
 \end{equation}  for all process matrix $W \in  \mathbf{W^{PROC}} $ if and only if 
 \begin{equation}
\sum_{ij} N_{ij}^{AB}  \in \mathbf{NS} (\A_I \otimes \A_O \otimes \B_I \otimes \B_O).
 \end{equation}

In this paper, we will consider different classes of process matrices. 
Initially, we define the subset of process matrices  known as free  objects in the resource theory of causal connection~\cite{milz2021resource}. Such process matrices will be  defined as follows. 

\begin{definition}
	We say that $W^{A || B} \in \mathbf{W^{PROC}} $ is  a  free process matrix if it satisfies the following condition 
	\begin{equation}\label{def-free}
	W^{A||B} = \rho_{\A_I \B_I } \otimes \Id_{\A_O \B_O}, 
	\end{equation}
	where $\rho_{\A_I \B_I } \in \Omega(\A_I \otimes \B_I)$ is an arbitrary quantum state and $ \Id_{\A_O \B_O}  \in \mathrm{L}(\A_O \otimes \B_O)$.  The set of all process matrices of this form will be denoted by $\mathbf{W^{A || B}}$.
\end{definition}

We often consider   process
matrices  corresponding to quantum combs\cite{chiribella2009theoretical}. For example, a quantum comb $A \prec B$  (see in Fig.~\ref{fig:comb}) shows that Alice's and Bob's operations are performed in causal order. This means that Bob cannot signal to Alice and the choice of
Bob’s instrument cannot influence the statistics
Alice records. Such process matrices are formally defined in the following way.

\begin{definition}
		We say that $W^{A \prec B} \in \mathbf{W^{PROC}} $ is  a   process matrix 	representing a quantum comb $A \prec B$ if it   satisfies the following conditions
	\begin{equation}\label{def-comb}\begin{split}
	W^{A \prec B} = W'_{\A_I\A_O\B_I} \otimes \Id_{\B_O}, \\
	\tr_{\B_I}  W'_{\A_I\A_O\B_I} = W''_{\A_I} \otimes \Id_{\A_O}.
	\end{split}
	\end{equation}
	The set of all process matrices  of this form will be denoted by $\mathbf{W^{A \prec B}}$.
\end{definition}

\begin{figure}[h!]
	\includegraphics[scale=1]{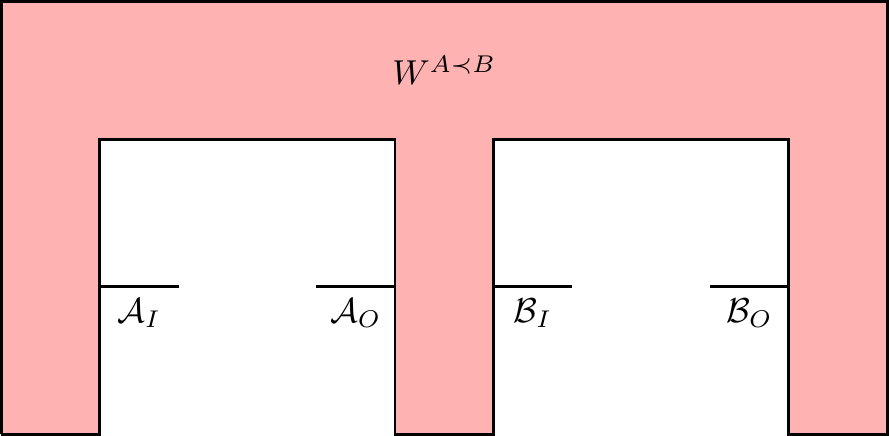}
	\caption{A schematic representation of a process matrix $W^{A \prec B}$ representing a quantum comb $A \prec B $. }
	\label{fig:comb}
\end{figure} 
One can easily observe that the set  $\mathbf{W^{A || B}}$ is an intersection of the sets $\mathbf{W^{A \prec B}}$ and $\mathbf{W^{B \prec A}}$. 
Finally,  the definition of the  set $\mathbf{W^{A \prec B}}$, together with $\mathbf{W^{B \prec A}}$ allow us to provide  their convex hull which is called as  causally separable process matrices. 
\begin{definition} 
 	We say that  $W^{SEP} \in \mathbf{W^{PROC}} $  is  a causally separable process matrix if it  is of the form
		\begin{equation}\label{def-sep}
	W^{\text{SEP}} = p W^{A \prec B } + (1-p) W^{B\prec A}, 
	\end{equation} where $W^{A \prec B} \in \mathbf{W^{A \prec B}}$, $ W^{B\prec A} \in \mathbf{W^{B \prec A}} $ for  some parameter $p \in [0,1]$.  The set of all  causally separable process matrices
	will be denoted by $\mathbf{W^{SEP}}$. 
\end{definition}

There are however process matrices that do not correspond
to a causally separable process and such process matrices are known as  causally
non-separable (CNS). The examples of such matrices were provided in~\cite{oreshkov2012quantum, oreshkov2016causal}.  
The set of all  causally non-separable process matrices
will be denoted by $\mathbf{W^{CNS}}$. 
In Fig.~\ref{fig:sets} we present a schematic plot of the sets of process matrices.

\begin{figure}[h!]
	\includegraphics[scale=3.5]{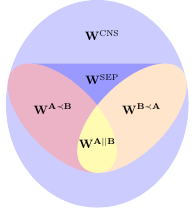}
	\caption{A schematic representation of the sets of process matrices $ \mathbf{W^{PROC}}$. }
	\label{fig:sets}
\end{figure}  

\section{Discrimination task}\label{sec-discrimination-prob}

This section presents the concept of discrimination between pairs of  process matrices.
It is worth emphasizing that the definition of a process matrix is a generalization of the concept of quantum states, channels, superchannels~\cite{gour2019comparison} and even generalized supermaps~\cite{jenvcova2012generalized, jenvcova2012extremality}.  The  task of discrimination between process matrices  poses a natural extension of 
discrimination of quantum states~\cite{helstrom1976quantum}, channels~\cite{watrous2018theory} or measurements~\cite{puchala2018strategies}. The process matrices discrimination task can be 
described  by the following scenario. 

 Let us consider two process matrices $W_0, W_1 \in  \mathbf{W^{PROC}} $. The classical description of process matrices $W_0, W_1 $ is  assumed to be known to the participating parties. We know that one of the process matrices, $W_0$ or $W_1$,  describes the actual correlation between Alice's and Bob's laboratories, but we do not know which one.
Our aim is to determine, with the highest possible probability,  which process matrix describes this correlation.  For this purpose, we construct a discrimination strategy $S$. 
In the general approach, such a  strategy $S$ is described by an instrument $S = \{ S_0, S_1 \}$. Due to the requirement given by Eq.~\eqref{eq:nonsignalling},
 the instrument $S$ must fulfill the condition  $S_0 + S_1 \in \mathbf{NS}(\A_I  \otimes \A_O  \otimes \B_I \otimes \B_O)$.
 The result of composing a process matrix $W$ with the discrimination strategy $S$ results in a classical label which can take values zero or one. 
  If the label zero occurs, we decide to choose that the correlation is given by $W_0$. Otherwise, we decide to choose $W_1$.  
 In this setting the  maximum success  probability  $p_{\text{succ}}  (W_0, W_1)  $  of correct discrimination  between two process matrices $W_0$ and $W_1$  can be expressed by  
\begin{equation}\label{eq:prob}
p_{\text{succ}} (W_0, W_1) = \frac{1}{2} \max_{S= \{S_0, S_1 \}} \left[ \tr(W_0 S_0) + \tr(W_1 S_1) \right]. 
\end{equation}
The following theorem  provides the optimal
probability of process matrices discrimination 
as a direct analogue of the Holevo--Helstrom theorem for quantum states and channels.

\begin{theorem}\label{th:holevo-helstrom}
	Let $W_0, W_1 \in \mathbf{W^{PROC}}$ be two process 
	matrices. For every choice of discrimination strategy $S = \{ S_0, S_1\}$, it 
	holds that
	\begin{equation}\label{eq-holevo}
	\begin{split}
&	\frac{1}{2} 	\tr(S_0W_0) + \frac{1}{2} \tr(S_1W_1)  \le  \\ &
	\frac{1}{2} + \frac{1}{4 } \max \left\{ \| \sqrt{N} (W_0 - W_1) \sqrt{N} 
	\|_1: N \in \mathbf{NS}(\A_I \otimes \A_O \otimes \B_I \otimes \B_O)  \right\} ,
	\end{split}
	\end{equation}
	where $\mathbf{NS}(\A_I \otimes \A_O \otimes \B_I \otimes \B_O)$ is the set of Choi matrices of non-signaling channels. Moreover, there exists a discrimination strategy $S$, which saturates the inequality Eq.~\eqref{eq-holevo}. 
\end{theorem} 	
	\begin{proof}

Let us define the sets \begin{equation}
\mathbf{A} \coloneqq \left\{ (S_0, S_1):S_0 + S_1 \in  \mathbf{NS}(\A_I \otimes \A_O \otimes \B_I \otimes \B_O) ,\,\, S_0, \, S_1 \ge 0 \right\}.
\end{equation}
		 and  
		 \begin{equation}
		 \begin{split}
		 \mathbf{B} \coloneqq \{ (\sqrt{N} Q_0 \sqrt{N}, \sqrt{N} Q_1 \sqrt{N}):  & \,  N \in  \mathbf{NS}(\A_I \otimes \A_O \otimes \B_I \otimes \B_O), \\ & 
		  Q_0,\, Q_1 \ge 0, \\ & Q_0 + Q_1 = \Id_{\A_I \A_O \B_I\B_O} \}.
		 \end{split}
		 \end{equation} We prove the equality between sets $\mathbf{A}$ and $\mathbf{B}$. To show $\mathbf{B} 
 \subseteq \mathbf{A}$, it is suffices to observe that $ \sqrt{N} Q_0 \sqrt{N} +  \sqrt{N} Q_1 \sqrt{N} \in  \mathbf{NS}(\A_I \otimes \A_O \otimes \B_I \otimes \B_O) $.		
To prove $\mathbf{A} \subseteq \mathbf{B}$ let us take $N \coloneqq S_0 + S_1$. It implies that
\begin{equation}
\Pi_{\text{im} (N)} = \sqrt{N}^{-1} N \sqrt{N}^{-1} = \sqrt{N}^{-1} S_0 
\sqrt{N}^{-1} + \sqrt{N}^{-1} S_1 \sqrt{N}^{-1}. 
\end{equation}
Let us fix $\widetilde{Q_0} \coloneqq \sqrt{N}^{-1} S_0 \sqrt{N}^{-1}$ and 
$\widetilde{Q_1} \coloneqq \sqrt{N}^{-1} S_1 \sqrt{N}^{-1}$. Then, we have $\Id_{\A_I \A_O \B_I\B_O} =  \Id_{\A_I \A_O \B_I\B_O} 
- \Pi_{\text{im} (N)} + \Pi_{\text{im} (N)} = \Id_{\A_I \A_O \B_I\B_O} - \Pi_{\text{im} (N)} + \widetilde{Q_0} + \widetilde{Q_1}$. Finally, it is suffices 
to take \begin{equation}\label{eq:Q0}
Q_0 \coloneqq  \Id_{\A_I \A_O \B_I\B_O} - \Pi_{\text{im} (N)} + \widetilde{Q_0}  
\end{equation}
 and $ Q_1 \coloneqq 
\widetilde{Q_1}$. It implies that $\mathbf{A} = \mathbf{B}$. 
In conclusion, we obtain
\begin{equation}
\begin{split}
&\frac{1}{2} \tr(S_0W_0) +  \frac{1}{2} \tr(S_1W_1)   \\& = \frac{1}{2} \tr \left( \sqrt{N} Q_0 \sqrt{N} W_0 \right) + \frac{1}{2} \tr 
\left( \sqrt{N} Q_1 \sqrt{N} W_1 \right)   \\& =   \frac{1}{2} \tr \left(  Q_0 
\sqrt{N} W_0 \sqrt{N} \right) + \frac{1}{2} \tr \left( Q_1 \sqrt{N} W_1 
\sqrt{N}  \right) \\&  \le \frac{1}{2} + \frac{1}{4 } \max \left\{ || \sqrt{N} (W_0 - W_1) \sqrt{N} 
||_1: N \in \mathbf{NS}(\A_I \otimes \A_O \otimes \B_I \otimes \B_O)  \right\}.
\end{split}
\end{equation}Moreover, from Holevo-Helstrom theorem~\cite{watrous2018theory} there exists a projective binary measurement $Q= \{ Q_0, Q_1\}  $ such that the last inequality is saturated,  which completes the proof.
 \end{proof}

\begin{corollary}
	The maximum   probability  $p_{\text{succ}}  (W_0, W_1)  $  of correct discrimination between two process matrices $W_0$ and $W_1$ is given by 
	\begin{equation}\label{eq:prob-exact}
	\begin{split}
		&p_{\text{succ}} (W_0, W_1) = \\& \frac{1}{2} + \frac{1}{4 } \max \left\{ \| \sqrt{N} (W_0 - W_1) \sqrt{N} 
	\|_1: N \in \mathbf{NS}(\A_I \otimes \A_O \otimes \B_I \otimes \B_O)  \right\}.
	\end{split} 
	\end{equation}

\end{corollary}

\begin{figure}[h!]
	\includegraphics[scale=0.6]{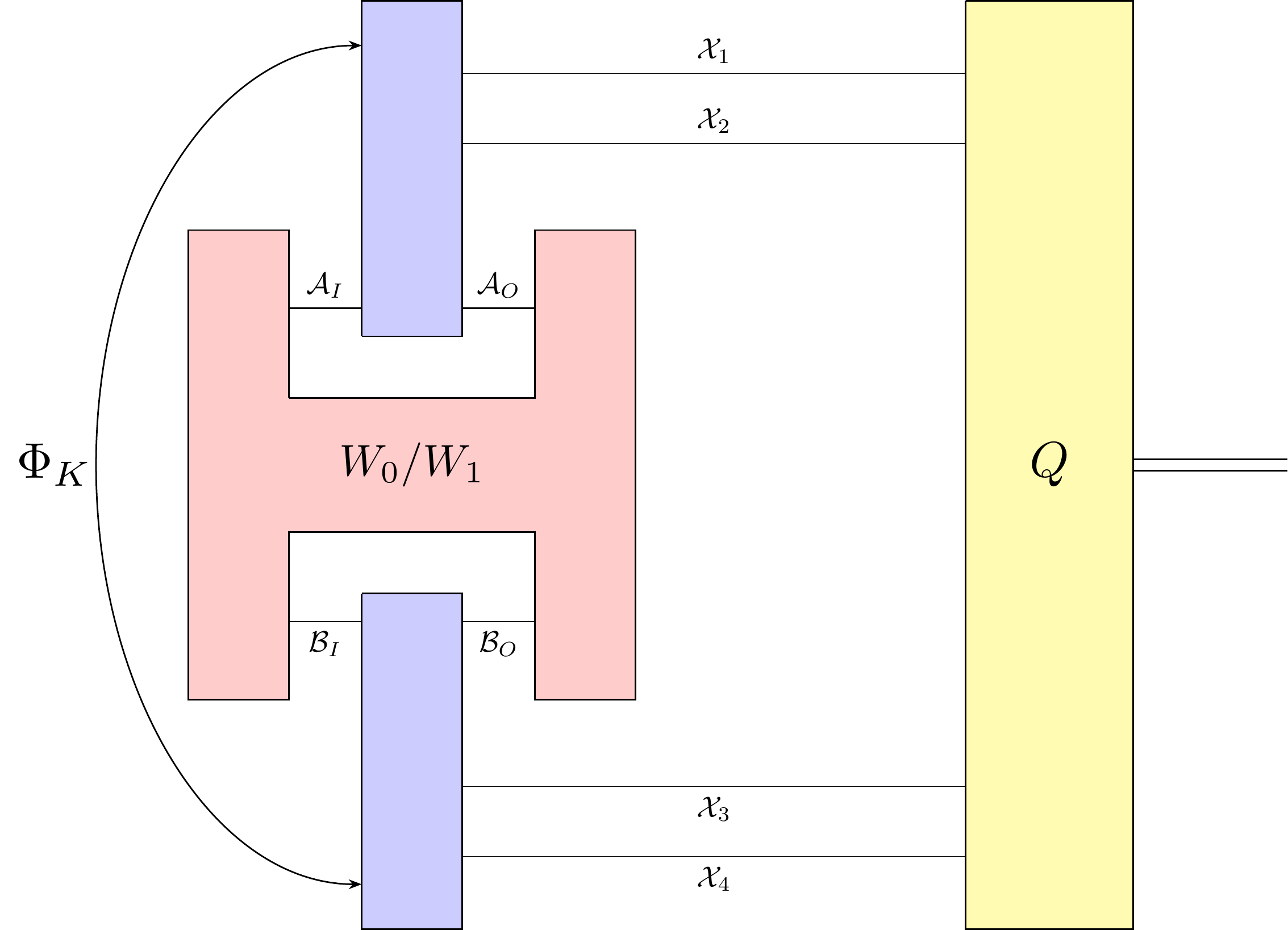}
	\caption{A schematic representation of the setup for
		distinguishing between process matrices $W_0$ and $W_1$. The discrimination strategy is constructed by using the quantum channel  $\Phi_K: \mathrm{L}(\A_I \otimes \B_I) \rightarrow \mathrm{L}(\A_O \otimes \B_O \otimes \XX_1 \otimes \XX_2 \otimes \XX_3 \otimes \XX_4 )$ and the  binary measurement $Q = \{ Q_0, Q_1\} $ defined in the proof of Theorem~\ref{th:holevo-helstrom}.   }
	\label{fig:realization}
\end{figure}

	As a valuable by-product of Theorem \ref{th:holevo-helstrom}, we receive a realization of  process matrices discrimination scheme. The schematic representation of this setup is presented in Fig.~\ref{fig:realization}. 
To distinguish 
the process matrices $W_0$ and  $W_1$, Alice and Bob prepare the strategy $S = \{ S_0, S_1 \}$ such that $S_0 + S_1 \in  \mathbf{NS}(\A_I \otimes \A_O \otimes \B_I \otimes \B_O)  $. 	To implement  it,  let us introduce complex Euclidean spaces  $\XX_1, \ldots, \XX_4$  such that $\dim\left(\bigotimes_{i=1}^4 \XX_i \right) = \dim(\A_I \otimes \A_O \otimes \B_I \otimes \B_O)$. 
	Alice and Bob prepare the  quantum channel $\Phi_K: \mathrm{L}(\A_I \otimes \B_I) \rightarrow \mathrm{L}(\A_O \otimes \B_O \otimes \XX_1 \otimes \XX_2 \otimes \XX_3 \otimes \XX_4 )$    with the  Choi matrix $K$ given by  
	\begin{equation}
	K = \left( \Id_{\XX_{1,2,3,4}} \otimes \sqrt{N} 
	\right)\left( |\Id \rangle\rangle \langle \langle \Id |  \right) 
	\left( \Id_{\XX_{1,2,3,4}} \otimes \sqrt{ N }\right), 
	\end{equation}
 where $| \Id \rangle \rangle \in \mathrm{L}(\XX_1 \otimes \XX_2 \otimes \XX_3 \otimes \XX_4  \otimes \A_I \otimes \A_O \otimes \B_I \otimes \B_O )$	and $N\in  \mathbf{NS}(\A_I \otimes \A_O \otimes \B_I \otimes \B_O)  $ maximizes the trace norm $\|\sqrt{N } (W_0 -W_1) \sqrt{N} \|_1$.
	It is worth noting that the quantum channel $\Phi_K$ is correctly defined due to the fact that $\tr_{\XX_{1,2,3,4}} K \in  \mathbf{NS}(\A_I \otimes \A_O \otimes \B_I \otimes \B_O)   $. Afterwards, 
	they  perform the binary measurement
	  $Q = \{ Q_0, Q_1\}$, where the effect  $Q_0 \in \mathrm{L}(\XX_1 \otimes \XX_2 \otimes \XX_3 \otimes \XX_4)$   is defined by Eq.~\eqref{eq:Q0}. Next, they decide which process matrix was used during the calculation  assuming $W_0$ if the measurement label is $0$. Otherwise, they assume $W_1$.

\section{Discrimination between different classes of process matrices }\label{sec-discrimination-example}
This section presents some examples of discrimination between different classes of process matrices. We begin our consideration with the problem of discrimination between two free process matrices $ \mathbf{W^{A || B}}$.
Next, we will consider   various cases of process matrices discrimination representing a quantum comb. 
First, we calculate exact probability of correct discrimination between two process matrices come from the same class $ \mathbf{W^{A \prec B}}$. Next, we study the discrimination task assuming that one of the process matrices is of the form  $ \mathbf{W^{A \prec B}}$ and the other one is of the form $ \mathbf{W^{B \prec A}}$. 
Finally, we 
construct a particular class  of process matrices which can be perfectly distinguished.

\subsection{Free process matrices} \label{sec-free}
The following consideration confirms an intuition that the task of discrimination between free process matrices reduces to the problem of discrimination between quantum states.

	From definition of $p_{\text{succ}}(W_0, W_1)$ we have 
	\begin{equation}
	\begin{split}
	p_{\text{succ}} \left(W_0, W_1\right) = \frac{1}{2} \max_{S=\{S_0, S_1\}} \left[ \tr\left(W_0 S_0\right) + \tr\left(W_1 S_1\right) \right], 
	\end{split} 
	\end{equation}
	Let $W_0$ and $W_1$ be two process matrices of the form 
	$W_0 =  \rho \otimes \Id_{\A_O \B_O}$ and $W_1 =  \sigma \, \otimes \, \Id_{\A_O \B_O}$, where $	\rho, \, \sigma \in \Omega(\A_I \otimes \B_I)$. 
		Then, $p_{\text{succ}}$  is exactly equal to
	\begin{equation}
	\begin{split}
	& \max_{S= \{ S_0, S_1\}} \left[\frac{1}{2}  \tr\left(W_0 S_0\right) +\frac{1}{2}  \tr\left(W_1 S_1\right) \right] = \\& \max_{S= \{ S_0, S_1\}} \left[ \frac{1}{2} \tr \left((\rho \otimes \Id) S_0 \right) + \frac{1}{2} \tr \left((\sigma \otimes \Id) S_1 \right) \right] = \\&  \max_{S= \{ S_0, S_1\}}  \left[ \frac{1}{2} \tr \left(\rho \tr_{\A_O\B_O} S_0 \right) + \frac{1}{2} \tr \left(\sigma \tr_{\A_O\B_O} S_1 \right) \right]. 
	\end{split} 
	\end{equation}
	Let us observe $\tr_{\A_O\B_O}S_0  +  \tr_{\A_O\B_O}S_1  = \Id_{\A_I \B_I}$. So, $ \{ \tr_{\A_O\B_O}S_0, \tr_{\A_O\B_O}S_1  \}$ is a binary measurement and therefore, from Holevo-Helstrom theorem for quantum states, we have 
	\begin{equation}
	\max_{S= \{ S_0, S_1\}}  \left[ \frac{1}{2} \tr \left(\rho \tr_{\A_O\B_O} S_0 \right) + \frac{1}{2} \tr \left(\sigma \tr_{\A_O\B_O} S_1 \right) \right] \le \frac{1}{2} + \frac{1}{4} \| \rho - \sigma \|_1. 
	\end{equation}
	Now, assume that $E = \{ E_0, E_1  \}$  is  the Holevo-Helstrom measurement (by taking $E_0$ and $E_1$ as positive and negative part of $\rho - \sigma$, respectively). Hence, we obtain \begin{equation}
	\max_{E= \{ E_0, E_1\}} \left[  \frac{1}{2} \tr \left(\rho E_0  \right) + \frac{1}{2} \tr \left(\sigma E_1  \right) \right]  =   \frac{1}{2} + \frac{1}{4} \| \rho - \sigma \|_1. 
	\end{equation} Observe, it is suffices to take $ S_0 \coloneqq E_0 \otimes \proj{0} \otimes \proj{0}  $ and   $ 
	S_1 \coloneqq E_1 \otimes \proj{0} \otimes \proj{0}  $.  Note that $ S_0 + S_1 = \Id_{\A_I \B_I}\otimes \proj{0} \otimes \proj{0} $ is non-signaling channel. Therefore, we have 
	\begin{equation}
	p_{\text{succ}} \left(W_0, W_1\right) = \frac{1}{2} + \frac{1}{4} \| \rho - \sigma \|_1.
	\end{equation}
	which completes the consideration.

Due to the above consideration, we obtain the following corollary. 
\begin{corollary}
Let $	\rho, \sigma \in \Omega(\A_I \otimes \B_I)$ be quantum states 
	and let  $W_0, W_1 \in \mathbf{W^{A || B}}$  be  two free process matrices of the form  $W_0 = \rho \otimes \Id$ and $  W_1=\sigma\otimes \Id$. Then, 
	\begin{equation}
	p_{\text{succ}} \left(W_0, W_1 \right)  = \frac{1}{2} + \frac{1}{4} \| \rho- \sigma\|_1.	\end{equation}
\end{corollary}
 
\subsection{Process matrices representing  quantum combs} \label{sec-comb}
Here, we will compare 
the probability of correct discrimination between two   process matrices being quantum combs of the form  $W_0^{A \prec B},W_1^{A \prec B }\in \mathbf{W^{A \prec B}}$ by using non-signalling strategy $S = \{ S_0, S_1 \}$ described by Eq.~\eqref{eq-p-succ}  or an adaptive strategy. 

 Before that, we will discuss the issue of adaptive strategy.
The most general strategy of quantum operations discrimination is known as an adaptive strategy~\cite{jenvcova2014base, chiribella2009theoretical}. An adaptive strategy is realized by a quantum tester~\cite{bisio2011quantum}. A schematic representation of this setup is presented in Fig.~\ref{fig:adaptive}.

Let us consider a quantum tester $\{ L_0, L_1 \}$. 
 The probability of correct discrimination  between $W_0^{A \prec B}$ and $W_1^{A \prec B}$ by using  an adaptive strategy is defined by equation
\begin{equation}\label{eq:adaptive-strategy}
p_{adapt }\left( W_0^{A \prec B}, W_1^{A \prec B} \right) \coloneqq \frac{1}{2} \max_{ \{ L_0, L_1\}}   L_0 * W_0 + L_1 * W_1. 
\end{equation}

\begin{figure}[h!]
	\includegraphics[scale=1]{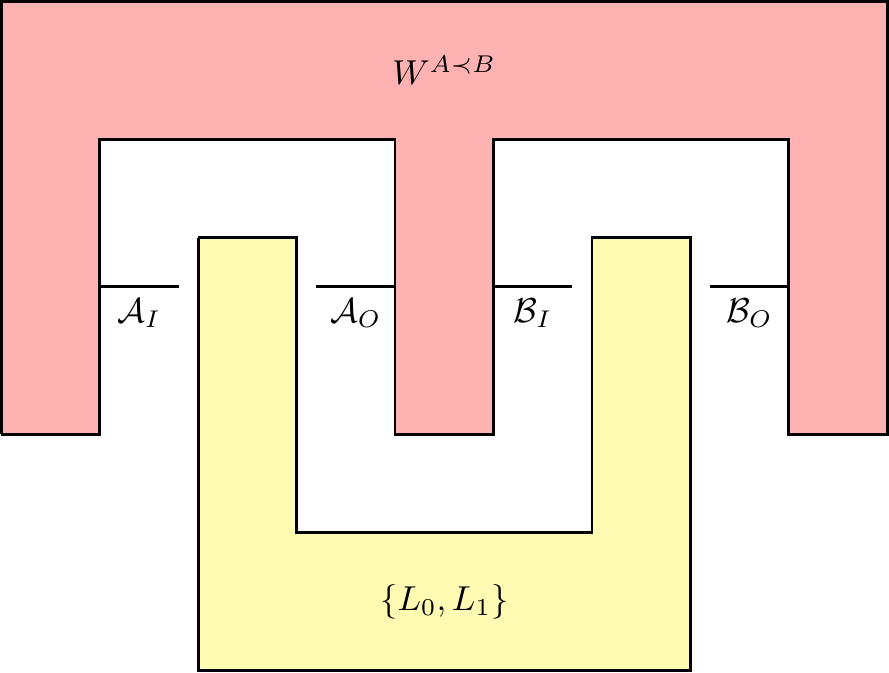}
	\caption{A schematic representation of an adaptive strategy discriminating two process matrices $W_0, W_1 \in W^{A \prec B}$ by using a quantum tester $\{ L_0, L_1\}$. }
	\label{fig:adaptive}
\end{figure}   

It turns out that we do not need adaptation in order to obtain the optimal probability os distinction. This is stated formally in the following theorem. 

\begin{theorem} Let  $W_0^{A \prec B}, W_1^{A \prec B} \in \mathbf{W^{A \prec B}}$ be two process matrices representing quantum combs $A \prec B$. Then, 
	\begin{equation}
	p_{succ} \left( W_0^{A \prec B}, W_1^{A \prec B} \right) = p_{adapt }\left( W_0^{A \prec B}, W_1^{A \prec B} \right). 
	\end{equation}
\end{theorem}
\begin{proof}
		For simplicity, we will omit  superscripts ($A \prec B$ and $B\prec A$). 
	The inequality $	p_{succ} \left( W_0, W_1 \right) \le  p_{adapt }\left( W_0, W_1 \right)  $ is trivial by observing that we calculate maximum value over a larger set.

	To show $	p_{succ} \left( W_0, W_1 \right) \ge  p_{adapt }\left( W_0, W_1 \right)  $,  let us consider the quantum tester $ \{ L_0, L_1\}$ which 
	maximizes Eq.~\eqref{eq:adaptive-strategy}, that means 
	\begin{equation}
	p_{adapt}\left( W_0, W_1\right) = \frac{1}{2} \left( L_0 * W_0 + L_1 * W_1 \right).
 	\end{equation}
 	Hence, from definition of $W^{A \prec B}$ we have 
 		\begin{equation}
 	W^{A \prec B} =  W'\otimes \Id_{\B_O}, 
 	\end{equation}
 	and then we obtain
 	\begin{equation}
 	 \frac{1}{2} \left( L_0 * W_0 + L_1 * W_1 \right) = 
 	 \frac{1}{2}  \tr \left( W'_0 \tr_{\B_O} L_0 + W'_1 \tr_{\B_O} L_1  \right).
 	\end{equation}
 	Observe that $\tr_{\B_O} (L_0 + L_1 ) = \Id_{\B_I} \otimes J$, where $J$ is a Choi matrix of a channel $\Phi_J: \mathrm{L}(\A_I) \rightarrow \mathrm{L}(\A_O )$. 
 	Let us define a strategy $S = \{S_0, S_1 \}$ such that 
 \begin{equation}
 \begin{split}
 & S_0 = \tr_{\B_O} L_0 \otimes \frac{\Id_{\B_O}}{\dim(\B_O)},\\
& S_1= \tr_{\B_O} L_1 \otimes \frac{\Id_{\B_O}}{\dim(\B_O)},\\
 & S = S_0 + S_1.
 \end{split}
 \end{equation}
 It easy to observe that $S\in \mathbf{NS}(\A_O \otimes \A_I \otimes \B_O \otimes \B_I)$. Then, we  have  
 \begin{equation}
 \frac{1}{2} \left( S_0 * W_0 + S_1 * W_1 \right) = 
 \frac{1}{2} \left( \tr \left( W'_0 \tr_{\B_O} L_0 + W'_1 \tr_{\B_O} L_1  \right) \right).
 \end{equation}
 It implies that 
\begin{equation}
p_{succ} \left( W_0, W_1\right) \ge  p_{adapt }\left( W_0,W_1\right), 
\end{equation}
which completes the proof. 
\end{proof}

\subsection{Process matrices of the form $W^{A \prec B}$ and $W^{B \prec A}$} \label{sec-perfect-discrim}

Now, we present some  results for discrimination task assuming the one of the process matrices if of the form  $ \mathbf{W^{A \prec B}}$ and the other one is of the form $ \mathbf{W^{B \prec A}}$.   We will construct a particular class of such process matrices for which the perfect discrimination is possible.

	Let us define a process matrix of the form \begin{equation}\label{eq:process-matrix}
W^{A \prec B} = \rho \otimes |U \rangle \rangle \langle \langle U |  \otimes \Id,
\end{equation} where  $  \rho \in \Omega(\A_I),  |U \rangle \rangle \langle \langle U | $ is the Choi matrix of a unitary channel
$\mathrm{Ad}_{U^\top}:\mathrm{L}(\A_O) \rightarrow \mathrm{L}(\mathcal{B}_I)$ 
of the form $\mathrm{Ad}_{U^\top}(X) = U^\top X \, \bar{U} $ and $\Id \in \mathrm{L}(\B_O)$. A schematic representation of this  process matrix we can see in Fig.~\ref{fig:perfectstrategy}.

\begin{figure}[h!]
	\centering
	\includegraphics[scale=1]{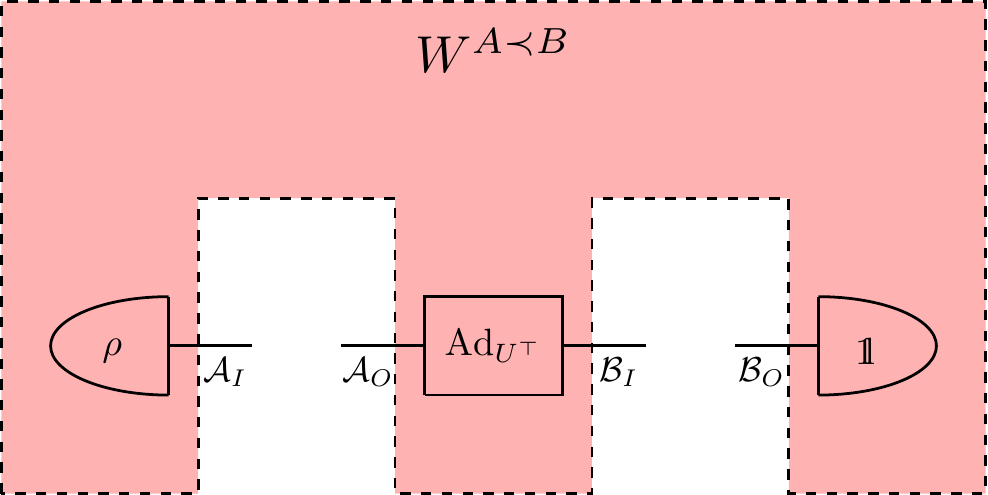}
	\caption{A schematic representation of  process matrix $W^{A \prec B} $ given by Eq.~\eqref{eq:process-matrix}.}
	\label{fig:perfectstrategy}
\end{figure} 
\begin{proposition} 
%
	Let $ W^{A \prec B} $ be a process matrix given by Eq.~\eqref{eq:process-matrix}. Let us define a process matrix $	W^{ B \prec A }$ of the form 
	\begin{equation}
	W^{ B \prec A } = P_{\pi} W^{A \prec B}  P_{\pi},
	\end{equation} where $P_{\pi}$ is the swap operator replacing the systems $\A_I \rightarrow \B_I $ and  $\A_O \rightarrow \B_O$. 
	 Then, the process matrix $W^{A \prec B}$ is perfectly distinguishable from $W^{B \prec A}$. 
\end{proposition}

\begin{proof}
Let us consider the process matrix given by Eq.~\eqref{eq:process-matrix} described by Fig.~\ref{fig:perfectstrategy}. 
 W.l.o.g. let  $d$ be a dimension of each of the systems. 
	Let  $\rho = \sum_{i=0}^{d-1} \lambda_i \proj{x_i}$, where $\lambda_i \ge 0 $ such that $\sum_i \lambda_i = 1$. Based on the spectral decomposition of $\rho$ we create the unitary matrix $V$ by taking  $i$-th eigenvector of $\rho$,  and the measurement $\Delta_V$ (in basis of $\rho$) given by
	\begin{equation}
	\Delta_V (X)= \sum_{i=0}^{d-1} \bra{x_i} X\ket{x_i } \proj{i} \otimes \proj{x_i}.  
	\end{equation}
	Let us also define 
	the permutation matrix $P_\sigma  = \sum_{i=0}^{d-1}\ketbra{x_{{i +1} \text{ mod d} }}{x_i}$
	corresponding to the permutation $\sigma = (0,1,\ldots,d-1)$. 
	
	Alice and Bob prepare theirs discrimination strategy.
	 Alice performs the local channel  (see Fig.~\ref{fig:alicestrategy}) given by 
	\begin{equation}\label{eq:alice-strategy}
	\Phi_{A}(\rho) = \left( \left(\mathcal{I} \otimes \mathrm{Ad}_{\bar{U}} \right) \circ \Delta_V\right)(\rho),
	\end{equation} 
	Meanwhile, Bob performs his local channel  (see Fig.~\ref{fig:bobstrategy}) given by \begin{equation}\label{eq:bob-strategy}
\Phi_{	B}(\rho) = \left(   \left(\mathcal{I} \otimes   \mathrm{Ad}_{\bar{U}} \right) \circ \Delta_V \circ  \mathrm{Ad}_{P_\sigma} \right)(\rho). 
	\end{equation} 
	\begin{figure}[h!]
		\includegraphics[scale=1]{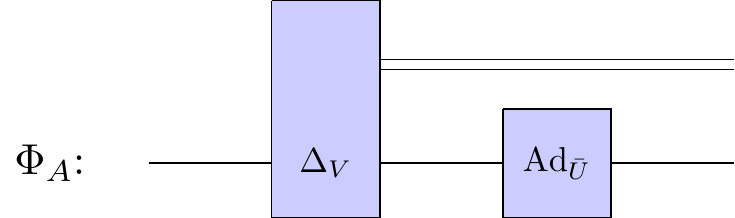}
		\caption{A schematic representation of Alice's discrimination strategy described by Eq.~\eqref{eq:alice-strategy}.}
		\label{fig:alicestrategy}
	\end{figure}
	\begin{figure}[h!]
		\includegraphics[scale=1]{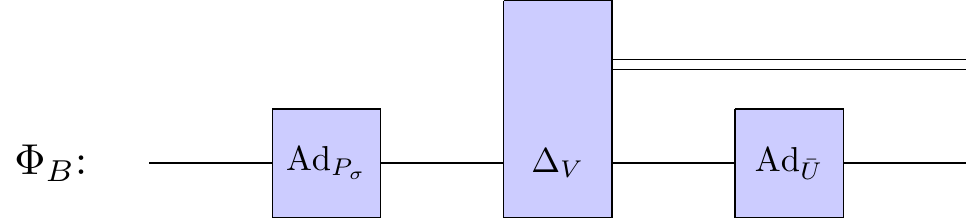}
		\caption{A schematic representation of  Bobs' discrimination strategy described by  Eq.~\eqref{eq:bob-strategy}.}
		\label{fig:bobstrategy}
	\end{figure} 

Let us consider the  case $A \prec B$. 
The output after Alice's action is described by \begin{equation}
\Phi_{A}(\rho) = \sum_{i} \lambda_i \proj{i} \otimes \bar{U} \proj{x_i} U^\top. 
\end{equation}
Next, we apply the quantum channel $\mathrm{Ad}_{U^\top}$ (see Fig.\ref{fig:perfectstrategy}), and hence we have \begin{equation}
(\mathcal{I}\otimes\mathrm{Ad}_{U^\top}) \circ \Phi_{A}(\rho)  = \sum_{i} \lambda_i \proj{i} \otimes\proj{x_i}. 
\end{equation}
In the next step, Bob applies his channel as follows
\begin{equation}
(\mathcal{I}  \otimes \Phi_{B}) \circ (\mathcal{I}\otimes\mathrm{Ad}_{U^\top}) \circ \Phi_{A}(\rho) = \sum_{i} \lambda_i \proj{i} \otimes \proj{i+1}\otimes\proj{x_{i+1}}. 
\end{equation}
Finally, we apply partial trace operation on the subspace $\B_O$  (see Fig.\ref{fig:perfectstrategy}), that means 
\begin{equation}
\tr_{\B_O} \left( \sum_{i} \lambda_i \proj{i} \otimes \proj{i+1}\otimes\proj{x_{i+1}} \right) = \sum_{i} \lambda_i \proj{i} \otimes \proj{i+1}. 
\end{equation}
So, the quantum state obtained after the discrimination scenario in the case $A \prec B$ is given by 
\begin{equation}
\sigma^{A \prec B} = \sum_{i} \lambda_i \proj{i} \otimes \proj{i+1}. 
\end{equation}
It implies that if Alice measures her system, she obtains the label $i$ with probability $\lambda_i$ whereas Bob obtains the label $(i+ 1) \mod  d$ with the same probability. 
On the other hand, considering the case 
$B \prec A$, then   the state obtained after the discrimination scenario is given by 
\begin{equation}
\sigma^{B \prec A} = \sum_{i} \lambda_i \proj{i+1} \otimes \proj{i+1}. 
\end{equation}
  So, Bob and Alice obtain  the same label $(i+1) \mod d $ with probability $\lambda_i$.
Then, the quantum channel $\Phi_K$  (realizing the discrimination strategy $S$) is created as a tensor product of Alice's and Bob's local channels, that means
$\Phi_K =  \Phi_A  \otimes \Phi_B$. 
	  Due to that they perform the binary measurement 
	$Q= \{ Q_0, Q_1 \}$, where the effect $Q_1$  is given by 
	$Q_1 = \sum_{i=0}^{d-1} \proj{i} \otimes \proj{i}$. Hence,  we have 
	\begin{equation}
	p_{\text{succ}}(W^{A \prec B}, W^{B \prec A}) = \frac{1}{2} \tr\left( \sigma^{A \prec B} Q_0\right) + \frac{1}{2} \tr\left( \sigma^{B \prec A} Q_1\right) = 1. 
 	\end{equation}
		In summary,  the process matrices $W^{A \prec B}$ and $W^{B \prec A}$ are perfectly distinguishable by Alice and Bob which completes the proof.

 \end{proof}

\section{SDP program for calculating the optimal probability of process matrices discrimination} \label{spd-prob}

In the standard approach, we would need
to compute the probability of correct discrimination between two process matrices $ W_0 $ and $W_1$. For this purpose, we  use  the
semidefinite programming (SDP). 
This section presents the SDP program for calculating the optimal probability of discrimination between $W_0$ and $W_1$. 

Recall that the maximum value   of such a probability can be noticed by
\begin{equation}\label{eq-p-succ}
p_{\text{succ}} (W_0, W_1) = \frac{1}{2} \max_{S= \{S_0, S_1 \}} \left[ \tr(W_0 S_0) + \tr(W_1 S_1) \right], 
\end{equation}
with requirement that the optimal strategy $S = \{S_0, S_1 \}$ is  a quantum instrument such that $S_0 + S_1 \in \mathbf{NS}(\A_I  \otimes \A_O  \otimes \B_I \otimes \B_O)$. 
Hence, we arrive at the primal and dual problems presented in the Program~\ref{sdp}.  
To optimize this problem we used the 
\texttt{Julia} 
programming language along with quantum package 
\texttt{QuantumInformation.jl}\cite{Gawron2018} and SDP optimization via SCS 
solver~\cite{ocpb:16, scs} with absolute convergence tolerance $10^{-5}$. The code is available on GitHub~\cite{code22}.

It may happen  that the values of 
primal and dual programs
are equal. This situation is called strong duality.
Slater's theorem provides the set of conditions which guarantee strong duality~\cite{watrous2018theory}. 
It can be shown that  Program~\ref{sdp} fulfills conditions of Slater's theorem (it is suffices to take $ Y_0,Y_1 = 0 $  and $\alpha > \frac{1}{2}  \max \{ \Lambda^{\text{max}}(W_0), \Lambda^{\text{max}}(W_1) \}$, where $\Lambda^{\text{max}}(X)$ is the maximum eigenvalue of $X$). Therefore, we can consider the primal and the dual problem equivalently. 
 \newpage

\begin{table}[ht!]
	\centerline{\underline{SDP program for calculating the optimal probability of discrimination between $W_0$ and $W_1$}} 
	\vspace{4mm}
	\centerline{\underline{Primal problem}} \vspace{-4mm}
	\begin{equation*}
	\begin{split}
	\text{maximize:}\quad &
	\frac{1}{2}  \Tr (W_0S_0) + \frac{1}{2}\Tr(W_1S_1) 
	\\[2mm]
	\text{subject to:}\quad & 
	\tr_{\A_O} (S_0 + S_1) = \frac{\Id_{\A_I}}{\dim(\A_I)} \otimes \tr_{\A_O\A_I} (S_0 + S_1),\\
	&\tr_{\B_O} (S_0 + S_1) = \frac{\Id_{\B_I}}{\dim(\B_I)} \otimes \tr_{\B_O\B_I} (S_0 + S_1), \\
	&\tr (S_0 + S_1) = \dim(\A_I)\dim(\B_I), \\
	& S_0 \in \mathrm{Pos}(\A_I  \otimes \A_O  \otimes \B_I \otimes \B_O), \\
	& S_1 \in \mathrm{Pos}(\A_I  \otimes \A_O  \otimes \B_I \otimes \B_O). 
	\end{split}
	\end{equation*}
	
	\hspace{2cm}
	
	\centerline{\underline{Dual problem}}\vspace{-4mm}
	\begin{equation*}
	\begin{split}
	\text{minimize:}\quad & 
	\alpha	\cdot \dim(\A_I)\dim(\B_I)  
	\\[2mm]
	\text{subject to:}\quad &
	\Id_{\A_O} \otimes Y_0 - \frac{\Id_{\A_I\A_O}}{\dim(\A_I)} \otimes \tr_{\A_I}(Y_0) + 
	\Id_{\B_O} \otimes Y_1 + \\& - \frac{\Id_{\B_I\B_O}}{\dim(\B_{I})} \otimes \tr_{\B_I}Y_1 + \alpha \cdot \Id_{\A_I\A_O\B_I\B_O}  \ge \frac{1}{2} W_0, \\
	&	\Id_{\A_O} \otimes Y_0 - \frac{\Id_{\A_I\A_O}}{\dim(\A_I)} \otimes \tr_{\A_I}(Y_0) + 
	\Id_{\B_O} \otimes Y_1 +  \\& - \frac{\Id_{\B_I\B_O}}{\dim(\B_{I})} \otimes \tr_{\B_I}Y_1 +  \alpha \cdot \Id_{\A_I\A_O\B_I\B_O} \ge \frac{1}{2} W_1, \\ 
	& Y_0 \in \text{ Herm}( \A_I \otimes \B_I \otimes \B_O), \\
	& Y_1 \in \text{ Herm}(\A_I \otimes \A_O   \otimes \B_I), \\ 
	& \alpha \in \R.
	\end{split}
	\end{equation*}
	
	\caption{Semidefinite program for maximizing the probability of correct discrimination between two process matrices $W_0$ and $W_1$.}
	\label{sdp}
\end{table}

\section{Distance between process matrices} \label{sec-sdp-distance}
In this section we  present the semidefinite programs for calculating   the distance  in trace norm between a given process matrix $W \in  \mathbf{W^{PROC}}$ and  
different subsets of process matrices, such that 
$ \mathbf{W^{A || B}}$ , $ \mathbf{W^{A \prec B}}$, $ \mathbf{W^{B  \prec A}}$ or $\mathbf{W^{SEP}}$.

For example, let us consider the case $ \mathbf{W^{A || B}}$. 
Theoretically, the distance between a process matrix $W$ and the set of free process matrices $ \mathbf{W^{A || B}}$  can be expressed by 
\begin{equation}\label{dist-free}
\begin{split}
& \text{dist}\left(W, \mathbf{W^{A || B}} \right)  =  \\ &\min_{\widetilde{W} \in \mathbf{W^{A || B}}} 
\max \left\{ \| \sqrt{N} (W - \widetilde{W}) \sqrt{N} 
\|_1: N \in \mathbf{NS}(\A_I \otimes \A_O \otimes \B_I \otimes \B_O)  \right\}.
\end{split}
\end{equation} Analogously, for the sets $\mathbf{W^{A \prec B}}, \mathbf{W^{B \prec A}} $ and $\mathbf{W^{SEP}}$ with the minimization condition  $\min_{\widetilde{W} \in \mathbf{W^{A \prec B}}}$, $\min_{\widetilde{W} \in \mathbf{W^{B \prec A}}}$, $\min_{\widetilde{W} \in \mathbf{W^{SEP}}}$, respectively.
Due to the results obtained from the previous section (see Program~\ref{sdp}) and Slater theorem
we are able to note the Eq.~\eqref{dist-free} to SDP problem presented in the  Program~\ref{sdp-distance}. 
We use the
SDP optimization via SCS 
solver~\cite{ocpb:16, scs} with absolute convergence tolerance $10^{-8}$ and relative convergence tolerance $10^{-8}$. 
The implementations of SDPs in the Julia language are available on GitHub \cite{code22}.

\begin{table}[ht!]
	
	\centerline{\underline{SDP calculating the distance between a process matrix $W$ and  the set $\Upsilon$.}} \vspace{-4mm}
	\begin{equation*}
	\begin{split}
	\text{minimize:}\quad & 
	4 	 \dim(\A_I)\dim(\B_I) \alpha -2  
	\\[2mm]
	\text{subject to:}\quad &
	\Id_{\A_O} \otimes Y_0 - \frac{\Id_{\A_I\A_O}}{\dim(\A_I)} \otimes \tr_{\A_I}(Y_0) + 
	\Id_{\B_O} \otimes Y_1 +  \\& - \frac{\Id_{\B_I\B_O}}{\dim(\B_{I})} \otimes \tr_{\B_I}Y_1 + \alpha \cdot \Id_{\A_I\A_O\B_I\B_O}  \ge \frac{1}{2} W, \\
	&	\Id_{\A_O} \otimes Y_0 - \frac{\Id_{\A_I\A_O}}{\dim(\A_I)} \otimes \tr_{\A_I}(Y_0) + 
	\Id_{\B_O} \otimes Y_1 +  \\& - \frac{\Id_{\B_I\B_O}}{\dim(\B_{I})} \otimes \tr_{\B_I}Y_1 +  \alpha \cdot \Id_{\A_I\A_O\B_I\B_O} \ge \frac{1}{2}W^{*},\\ 
	& W^{*} \in \Upsilon, \\
	& Y_0 \in \text{ Herm}( \A_I \otimes \B_I \otimes \B_O), \\
	& Y_1 \in \text{ Herm}(\A_I \otimes \A_O   \otimes \B_I), \\ 
	& \alpha \in \R.
	\end{split}
	\end{equation*}
	
	\caption{Semidefinite program for computation  the distance between a process matrix $W$ and $\Upsilon$, which can be one of the set $ \mathbf{W^{A || B }},  \mathbf{W^{A \prec B}}, \mathbf{W^{B \prec A }}$ or $\mathbf{W^{SEP}} $. Depending on the selected set
		we include additional constrains to SDP described by Eq.~\eqref{def-free} for $\mathbf{W^{A || B }}$, Eq.~\eqref{def-comb}  for  $\mathbf{W^{A \prec B}}$ and $\mathbf{W^{B \prec A}}$ or Eq.~\eqref{def-sep} for $\mathbf{W^{SEP}} $.  }
	\label{sdp-distance}
\end{table}

  \subsection{Example}
Let $\A_I=\A_O=\B_I=\B_O=\C^2$. 
Let us consider 
a causally non-separable process matrix comes from \cite{oreshkov2012quantum} of the form
\begin{equation}\label{eq:cns}
W^{\text{CNS}} = \frac{1}{4} \left[\Id_{\A_I\A_O\B_I\B_O} + \frac{1}{\sqrt{2}} \left( \sigma_z^{\A_O} \sigma_z^{\B_I}  \otimes \Id_{\A_I\B_O} +  \sigma_z^{\A_I} \sigma_x^{\B_I} \sigma_z^{\B_O} \otimes \Id_{\A_O} \right) \right], 
\end{equation}
where $\sigma_x^\XX, \sigma_z^\XX$ are Pauli matrices on space $\mathrm{L}(\XX)$.  
 We have calculated the distance in trace norm between $W^{\text{CNS}}$ and  different subset of process matrices. Finally, we obtain
\begin{equation}
\text{dist}\left( W^{\text{CNS}},\mathbf{W^{A || B }}\right)\approx 1.00000001\approx 1,
\end{equation} 
\begin{equation}
\text{dist}\left( W^{\text{CNS}}, \mathbf{W^{A \prec B }} \right) \approx 0.7071068 \approx \frac{\sqrt{2}}{2},
\end{equation}
\begin{equation}
\text{dist}\left( W^{\text{CNS}}, \mathbf{W^{B \prec A }} \right) \approx 0.7071068 \approx \frac{\sqrt{2}}{2},
\end{equation}
\begin{equation}
\text{dist}\left( W^{\text{CNS}}, \mathbf{W^{SEP }} \right) \approx 0.2928932 \approx 1-\frac{\sqrt{2}}{2}. 
\end{equation}
The numerical computations  give us some intuition  about the geometry of the set of process matrices. Those results are presented in Fig.~\ref{fig:distance}. 
Moreover, by using $W^{\text{CNS}} $ given by Eq.~\eqref{eq:cns} it can be shown that the set of all causally non-separable process matrices is not convex. 
To show this fact, it suffices to observe that for every $\sigma_i, \sigma_j, \sigma_k, \sigma_l \in \{ \sigma_x, \sigma_y, \sigma_z, \Id \}$ the following equation holds \begin{equation}\label{eq:sigma}
\left(\sigma_i \otimes \sigma_j \otimes \sigma_k \otimes \sigma_l \right) W^{\text{CNS}} ( \sigma_i \otimes \sigma_j \otimes \sigma_k \otimes \sigma_l)^\dagger \in \mathbf{W^{CNS}}. 
\end{equation} 
Simultaneously, the average of the process matrices of the form Eq.~\eqref{eq:sigma} distributed uniformly states $\frac{1}{4} \Id_{\C^{16}} $, however  $ \frac{1}{4} \Id_{\C^{16}} \not\in \mathbf{W^{CNS}} $. It implies that the set $\mathbf{W^{CNS}}$ is not convex. 

\begin{figure}[h!]
	\includegraphics[scale=3.5]{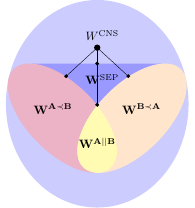}
	\caption{A schematic representation of the distances between $W^{\text{CNS}}$ defined in Eq.~\eqref{eq:cns} and the sets $\mathbf{W^{A || B }}$, $\mathbf{W^{A \prec B }}$, $\mathbf{W^{B \prec A }}$ and $\mathbf{W^{SEP}}$. }
	\label{fig:distance}
\end{figure}

\section{Convex cone structure theory} \label{sec-convex}
From geometrical point of view, we present an alternative way to 
derive of Eq.~\eqref{eq:prob-exact}. It turns out that the task of process matrices discrimination is strictly connected with the convex cone structure theory.  To keep this work self-consistent, the details of convex cone structure theory 
are presented in Appendix~\ref{app:convex}.

Let $\VV$ be a finite dimensional real vector space with a proper cone $\CC \subset \VV$. 
A base $B$ of the proper cone $\CC$ is a compact convex subset $B \subset \CC$
such that
each nonzero element  $c \in \mathcal{C} $ has a unique representation in the form $c = \alpha \cdot b$, where $\alpha > 0 $ and $b \in B$. 
The corresponding base norm in $\VV$ is defined by 
\begin{equation}
||x||_B = \{ \alpha + \beta, x = \alpha b_1 - \beta b_2, \alpha, \beta \ge 0, b_1, b_2 \in B \}. 
\end{equation}	
From \cite[Corollary 2]{jenvcova2014base}  the author showed that  the base norm can be written as 
\begin{equation}
|| x ||_B =  \sup_{\widetilde{b}\in \widetilde{B}} ||\widetilde{b}^{1/2}  x \widetilde{b}^{1/2}  ||_1, 
\end{equation}
where \begin{equation}
\widetilde{B} \coloneqq \{ Y  \in \CC: \tr \left( XY \right) = 1, \forall X \in B  \}. 
\end{equation}

\subsection{Convex cone structure of process matrices set}

Let  $\VV$  be a Hilbert space  given by 
\begin{equation}
\VV = \mathrm{Herm}(\A_I\otimes \A_O \otimes \B_I \otimes \B_O),
\end{equation}
with proper cone  \begin{equation}
\CC = \{ W \in \VV: W \in \mathrm{Pos}(\A_I\otimes \A_O \otimes \B_I \otimes \B_O)  \}.
\end{equation}
Consider the linear subspace $\s \subset \VV$  given by \begin{equation} \begin{split}
\s  = \{W \in \VV :  &\prescript{}{B_IB_O}{W} = \prescript{}{A_OB_IB_O}{W}, \\ 
&\prescript{}{A_IA_O}{W} = \prescript{}{A_OA_IB_O}{W}, \\
&\prescript{}{A_OB_O}{W} = \prescript{}{B_O}{W} + \prescript{}{A_O}{W} - W  \}
\end{split}
\end{equation}
together with its proper cone $\CC_\s$. Observe that if we fix  trace of $W \in \CC_\s$ such that $ \tr(W) = \dim(\A_O) \cdot \dim(\B_O)$,  we achieve the 
set of all process matrices $ \mathbf{W^{PROC}}$. And then, $ \mathbf{W^{PROC}}$ is a base of 
$\CC_\s$. 

\begin{proposition}\label{no-signalling}
	Let $\mathbf{W^{PROC}}$ be the set of  process matrices. Then,  the set $\widetilde{\mathbf{W^{PROC}}}$ 
	is  determined by 
	\begin{equation}
	\widetilde{\mathbf{W^{PROC}}} =  \mathbf{NS}(\A_I\otimes \A_O \otimes \B_I \otimes \B_O).
	\end{equation}
\end{proposition}
\begin{proof}
	We want to prove that \begin{equation}
	\tr \left( X W \right) = 1 
	\,\, \text{for all } W \in \mathbf{W^{PROC}}
	\iff X \in  \mathbf{NS}(\A_I\otimes \A_O \otimes \B_I \otimes \B_O).
	\end{equation}   Let us first take $X \in \mathbf{NS}(\A_I\otimes \A_O \otimes \B_I \otimes \B_O) $.  Then, from \cite[Lemma 1]{chiribella2013quantum}, 
	we note \begin{equation}
	X = \sum_{i} \lambda_i A_i \otimes 
	B_i,
	\end{equation} where $A_i \in \mathbf{NS}(\A_I\otimes \A_O) $,  $B_i \in \mathbf{NS}(\B_I\otimes \B_O)$ and $\lambda_i \in \R$ such that $\sum_{i} \lambda_i = 1$. From 
	definition of process matrix and linearity we obtain \begin{equation}
	\tr \left( X 
	W \right) = \tr \left( \sum_i \lambda_i (A_i \otimes B_i)  W \right) = 
	\sum_i \lambda_i \tr \left(( A_i \otimes B_i ) W \right) = \sum_{i} 
	\lambda_i 
	= 1. 	\end{equation}
	To prove opposite implication, let us take $W = 
	\Id_{\A_O} \otimes J$, where $J$ is the Choi 
	matrix of quantum channel $\Phi_J: \mathrm{L}(\B_I \otimes \B_O) \rightarrow \mathrm{L} (\A_I)$. Then, we have
	\begin{equation}
	1=\tr \left( 	W X \right) = \tr\left(  (\Id_{\A_O} \otimes J) X  
	\right) = 
	\tr \left( J \tr_{\A_O} X   \right).
	\end{equation} 
	From \cite{lewandowska2020optimal} we have 
	\begin{equation}
	\tr_{\A_O} X = \Id_{\A_I} \otimes P, 
	\end{equation}
	where $P \in \mathrm{Pos}(\B_I \otimes \B_O)$.	
	Similarly, if we take $W \coloneqq \Id_{\B_O} \otimes K$, where $K$  is the  Choi 
	matrix of a quantum channel $\Phi_K: \mathrm{L}(\A_I \otimes \A_O) \rightarrow \mathrm{L} (\B_I)$,  we 
	obtain \begin{equation}
	\tr_{\B_O} X = \Id_{\B_I} \otimes P, 
	\end{equation}
	where $P \in \mathrm{Pos}(\A_I \otimes \A_O)$. 
	It implies that $X \in  \mathbf{NS}(\A_I\otimes \A_O \otimes \B_I \otimes \B_O)$,  which completes the proof. 
\end{proof}
Due to Proposition~\ref{no-signalling},  we immediately
obtain
the following corollary. 
\begin{corollary}
	The base norm  $||\cdot||_\mathbf{W^{PROC}}$ between two process matrices 
	$W_1, W_2 \in \mathbf{W^{PROC}}$ can be expressed as 
	\begin{equation}
	|| W_1 - W_2 ||_{\mathbf{W^{PROC}} } 
	= \max \{ ||\sqrt{N} (W_1 - W_2) \sqrt{N} ||_1: N \in \mathbf{NS}(\A_I\otimes \A_O \otimes \B_I \otimes \B_O) \}.
	\end{equation}
\end{corollary}

 \section{Conclusion and discussion} \label{sec-conclusion}
 
 In this work, we studied  the problem of  single shot
 discrimination between process matrices. 
 Our aim was to provide an exact expression for the optimal probability of correct distinction and quantify it in terms of the
 trace norm.  This value was maximized over all Choi operators of non-signaling channels and 
 and poses direct analogues to the Holevo-Helstrom theorem for quantum channels.
 In addition, we have presented an alternative way to achieve this expression by using the convex cone structure theory.
 As a valuable by-product, we have also found  the optimal realization of the discrimination task for process matrices that use such non-signalling channels. 
 Additionally, we expressed the discrimination task as semidefinite programming (SDP). 
 Due to that, we have created SDP calculating the distance between process matrices and 
 we expressed it in terms of the
 trace norm.  
 Moreover, we found
 an analytical result for discrimination of free process matrices. It turns out that the task of discrimination between free process matrices can be reduced to the task of discrimination between quantum states. 
 Next, we consider the problem of discrimination for process matrices corresponding to quantum combs. 
 We have studied which strategy, adaptive or non-signalling, 
 should be used during the discrimination task.
 We proved that no matter which strategy you choose, the optimal probability of distinguishing two process matrices being a quantum comb is the same.
So, it turned out that we do not need to use some unknown additional processing in this case.
 Finally,   we discovered a particular class of 
 process matrices having opposite causal order, which can be distinguished perfectly. 
  This work paves the way toward a complete description of
  necessary
  and sufficient criterion for perfect discrimination between process matrices.
  Moreover, it poses a starting point to fully describe the geometry of the set of process matrices, particularly causally non-separable process matrices.

\section*{Acknowledgements}

This work was supported by the project ,,Near-term quantum computers Challenges, 
optimal implementations and applications'' under
Grant Number POIR.04.04.00-00-17C1/18-00, which is carried out within the 
Team-Net programme of the
Foundation for Polish Science co-financed by the European Union under the 
European Regional
Development Fund and  SONATA BIS grant number 2016/22/E/ST6/00062.

PL is a holder of European Union scholarship through the European Social Fund, 
grant InterPOWER (POWR.03.05.00-00-Z305).

\bibliographystyle{ieeetr}
\bibliography{discrimination}

\appendix

\section{Convex cone structures}\label{app:convex}

To keep this work self-consistent we
present in this appendix  basic definitions and properties about convex cone structure theory. 

Suppose $\XX$ is a finite dimensional real vector space   and $\CC \subset \XX$ 
is  a  closed  convex  cone. We  assume  that $\CC$ is pointed,  that means  $\CC 
\cap 
-\CC = \{ 0 \}$.
A closed pointed convex cone is in one-to-one correspondence
with partial order in $\XX$, by $
x \ge y \iff x-y  \in \CC
$
for each $x,y \in \XX$. If we additionally assume that  the cone is generating,  that is for each $ x \in \XX$ there exists $ u, 
w 
\in \CC$ such that $
x = u - w$, then
a nonempty set $\CC \subseteq \XX$ satisfying all above properties will be called a  
proper cone in space $\XX$. 
Let  $\XX^*$  be a  dual space with duality $\braket{\cdot}{\cdot}$.  Then, 
we  introduce a partial order in $\XX^*$ as well with dual cone $	\CC^* = 
\{f \in \XX^*: \braket{f}{z} \ge 0, \forall z \in C \}.$
Observe that the cone $\CC^*$ is also closed and convex. Moreover, if $\CC$ is generating in 
space $\XX$, then $\CC^*$ is pointed, so
we can introduce  the  partial order in $\XX^*$ given by
\begin{equation}
f \ge g \iff f-g \in \CC^*
\end{equation}
for all $f,g \in \XX^*$. 

Next,  consider a linear space  with fixed inner product. If $\XX$ is an inner product space, then the Riesz representation theorem~\cite{rudin1991functional} 
holds that 
the inner product determines an  isomorphism between $\XX$ and $\XX^*$. 
Therefore, the cone $\CC$ is equal to $\CC^*$.

  An interior point $e \in \mathrm{int}(\CC)$ of a cone $\CC$ is called an 
order unit if for each $ x \in \XX$, there exists  $ \lambda >0 $ such that $
\lambda e - x \in \CC$. Whereas, a base of $\CC$ is defined as compact and convex 
subset 
$B \subset \CC$ such that for every $z \in \CC\setminus\{0\}$, there exists 
unique $t > 0$  and an element $b \in B$ such that $
z = tb. $ 
It can be shown that the set  \begin{equation}
B= \{z \in \CC: \braket{e}{z} = 1 \} 
\end{equation} is the 
base of $\CC$ 
	(determined by element $e$) if and only if an element $e$  is an order unit 
	and $e \in \mathrm{int}\left(\CC^*\right)$.
	Finally, we define the base norm as
\begin{equation}
||x||_B = \{ \alpha + \beta, x = \alpha b_1 - \beta b_2, \alpha, \beta \ge 0, b_1, b_2 \in B \}. 
\end{equation}
 It can be shown~\cite{jenvcova2014base}  that the base norm is expressed as 
\begin{equation}
||x||_B = 
\sup_{\widetilde{b} \in \widetilde{B} } || \widetilde{b}^{1/2} x 
\widetilde{b}^{1/2} ||_1, 
\end{equation}
where 
$
\widetilde{B} =  \{ \widetilde{b} \in \CC : \tr(b\widetilde{b} ) = 
1, \forall b \in B \}. 
$

\end{document}